\documentclass{cccg12}

\usepackage{graphicx,amssymb,amsmath,dsfont}

\usepackage[small,bf,hang]{caption}
\usepackage{url}

\usepackage{verbatim}

\title{Computing the Coverage of an Opaque Forest \thanks{This work was supported by NSERC}}
\author{Alexis Beingessner \and Michiel Smid \thanks{School of Computer Science, Carleton University}}

\begin{document}
\thispagestyle{empty}
\maketitle


\begin{abstract}
We consider the problem of taking an opaque forest and determining the regions that are covered by it. We provide a tight upper bound on the complexity of this problem, and an algorithm for computing this area, which is worst-case optimal.
\end{abstract}


\section{Introduction}
Let a \emph{region} be any bounded, closed, and connected set of points in $\mathds{R}^2$. Then a \emph{barrier}, or \emph{opaque forest}, of a finite set $R$ of regions, is any finite set $B$ of closed and bounded line segments, such that for any line $\ell$: if $\ell$ intersects $R$ then $\ell$ also intersects $B$. A previously studied problem is as follows: given some set $R$ or regions, compute a barrier $B$ such that the length of all the segments in $B$, $|B|$, is minimal. The exact solution to this problem is not known even for specific cases, such as when $R$ is a unit square. The best known bounds for this instance of the problem are $2 \leq |B| \leq \sqrt{2} + \sqrt{6}/2$. \cite{Dumitrescu2010}

The general problem of computing a minimal barrier for a given set of regions is a very difficult one. Currently there are no proven algorithms for computing this precisely, nor even known solutions for specific cases. For the internally optimal barrier, there is also no known algorithm. However, by further restricting the problem, it is reducible to well studied problems. If the internally optimal barrier is restricted to a single connected component, then this is easily reducible to the \emph{Minimal Steiner Tree Problem}. If the barrier is further restricted to a single polygonal chain, then the problem is reducible to the \emph{Travelling Salesman Problem}. Both of these problems are known to be NP-Hard in general, but can be much more easily computed or approximated when the input points are in convex position, which is the case for this problem \cite{Dumitrescu2010}.

In this paper we consider the following problem: given some barrier $B$, compute a maximal set $R$ of regions such that $B$ is a barrier for $R$. More precisely, given a set $B$ of $n$ line segments, compute $R(B) = \{p \in \mathds{R}^2:$ every line through $p$ intersects $B\}$. We say that $R(B)$ is the \emph{coverage} of $B$.

We give an algorithm that computes the coverage of an opaque forest in
$O(n^4)$ time. We also provide an example of an opaque forest whose
coverage has size $\Omega(n^4)$. Thus, our algorithm is worst-case
optimal.

\section{Maximal Regions}

 Let a \emph{maximal region} of a set $P$ of points be a region
   $R$ such that for every point $p$ in $R$, there exists an open
   ball $A$ centered at $p$ such that $A \cap R = A \cap P$.

\begin{lemma}
\label{lemma:hasNoLines}
If a maximal region of $R(B)$ is a line segment, then that line segment is part of $B$.
\end{lemma}

\begin{proof}
Assume this is not the case. Then there is some line segment $S \in R(B)$ that is a maximal region, but is not in $B$. Therefore all lines that pass through a point $p$ in $S$ intersect $B$, and there exists an open ball $A$ of points around $p$ such that every point $q$ in $A$ that is not in $S$ has a line $\ell$ through it which does not intersect $B$.

Consider such a point $q$. The line $\ell$ through $q$ that does not intersect $B$ cannot intersect $S$, or else the points it intersects in $S$ are not actually in $R(B)$. We can select a point $q'$ such that it is arbitrarily close to $p$, and the line $\ell'$ must therefore become ever more parallel to the line $S$ lays on to avoid intersection. Therefore it must be the case that the line collinear with $S$ intersects $B$, but the line $\ell'$ that is parallel to $S$ and arbitrarily close to it does not. Therefore, there must exist some line segment $S' \in B$ that is parallel to $S$. Further, there must be some opaque forests to the left and right of $S$ that do not meet each other or $S'$, or else $\ell$ can pass through $S$. Therefore, there is a space for parallel lines to the left and right of $S$. However, this implies that there is a line $\ell''$ that enters through one space and exits through the other which does not intersect $B$ but passes through $S$, which means there are points in $S$ which are not in $R(B)$. If this were not the case, then $\ell'$ would intersect $B$. So we have a contradiction, therefore if $S$ is in $R(B)$, $S$ is in $B$.\end{proof}

\begin{figure}[ht]
  \centering
  \includegraphics[scale=0.4]{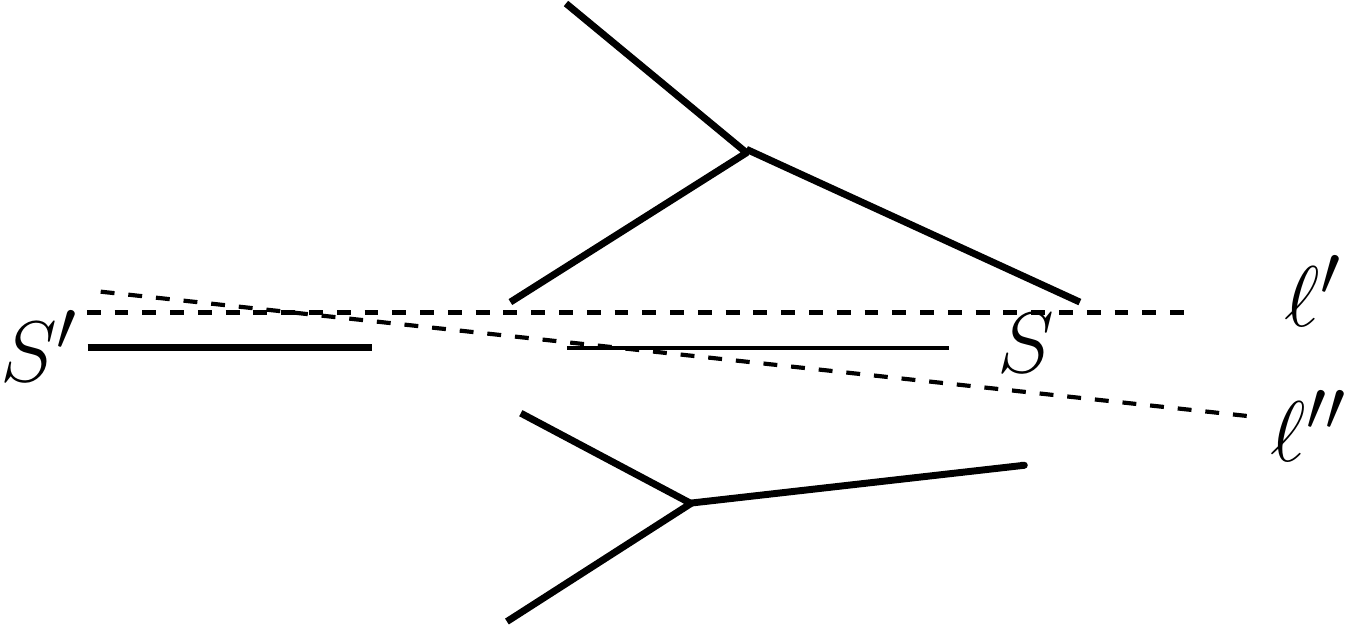}
  \caption{There is a line $\ell''$ which does not intersect $B$ but passes through $S$}
  \label{fig:noSegments}
\end{figure}

\begin{lemma}
\label{lemma:hasPoints}
$R(B)$ may contain maximal regions that are single points, but are not part of $B$.
\end{lemma}

\begin{proof}
Consider the construction of three line segments found in Figure 2. 

\begin{figure}[ht]
  \centering
  \includegraphics[scale=0.4]{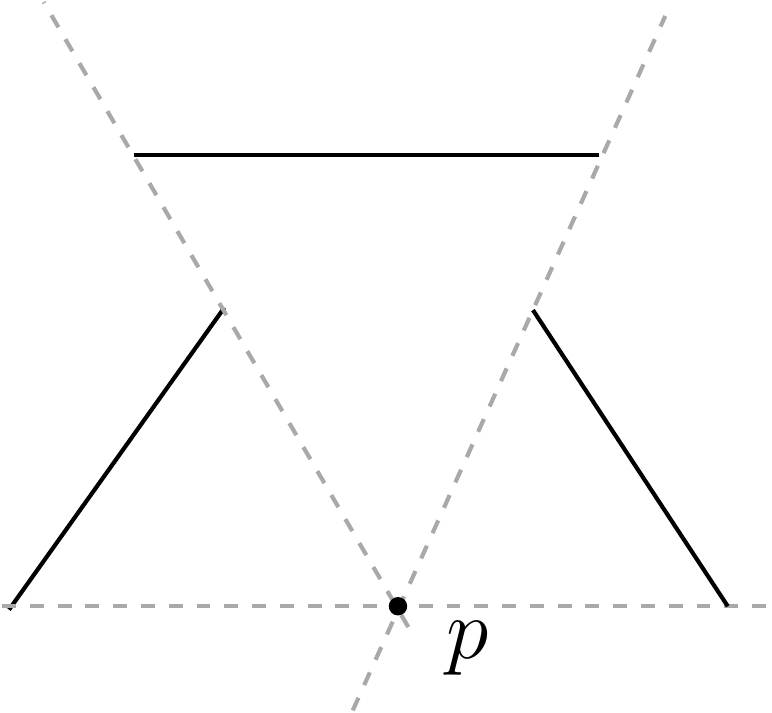}
  \caption{A construction that creates a maximal region that is exactly one point}
  \label{fig:pointGenerator}
\end{figure}

$p$ is not part of $B$. Every line that passes through $p$ intersects $B$, so $p \in R(B)$. Yet there exists an open ball of points centred at $p$ such that every point in this ball except for $p$ has a line through it that does not intersect $B$. Therefore $p$ is a maximal region of $R(B)$.\end{proof}


\section{Clear and Blocked Points}

Let a \emph{blocked point} be a point $p$ with respect to some barrier $B$ such that for every line $\ell$ which passes through $p$, $\ell$ intersects $B$. Then a \emph{clear point} is a point which is not blocked. Every point of $B$ is a blocked point. Moreover, $R(B)$ is the set of all blocked points with respect to $B$, and the complement $\overline{R(B)}$ or $R(B)$ is the set of all clear points. 

\begin{theorem}
\label{thm:tangency-boundary}
For every barrier $B$, each maximal region $C \subseteq R(B)$ is the intersection of halfplanes defined by lines that pass through two vertices of $B$.
\end{theorem}

\begin{proof}

Assume there exists some line $\ell$ which is tangent to the boundary of a maximal region $C \subseteq R(B)$, but $\ell$ does not touch $B$. Then, because the complement of $B$ is an open set, $\ell$ can be translated to intersect $C$ without intersecting $B$. However that would mean $C$ contains clear points, which is a contradiction. Therefore, $\ell$ must be tangent to $B$ at at least one point. Now assume $\ell$ is tangent to $B$ at exactly one point. Then $\ell$ can still be rotated around the point of tangency, once more intersecting $C$. This once more contradicts the fact that $C$ is a subset of $R(B)$. Therefore $\ell$ must be tangent to at least two points of $B$. Further, since $B$ is a set of line segments, only the end points of these segments need be considered, as tangency to a line segment is simply tangency to its two end points.\end{proof}

\begin{figure}[ht]
  \centering
  \includegraphics[scale=0.4]{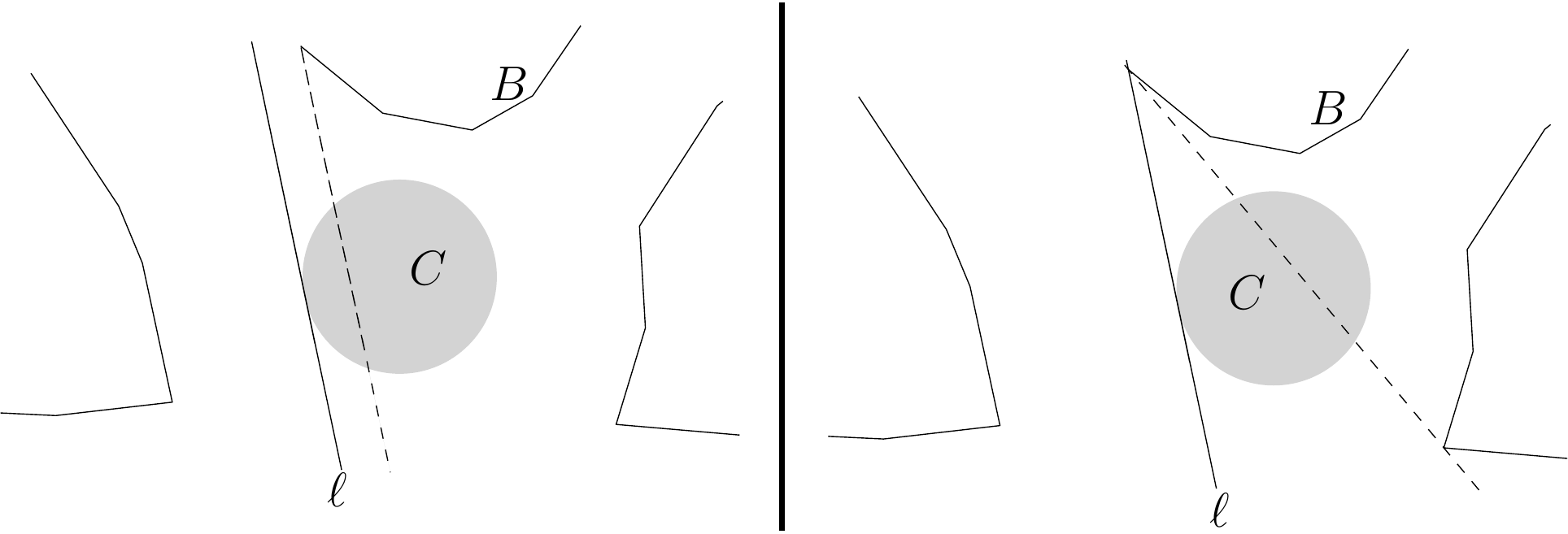}
  \caption{The line $\ell$ must be tangent to $B$ at two vertices, if it defines the boundary of part of $R(B)$}
  \label{fig:halfPlaneBounding}
\end{figure}

Remark that this also implies that we need only finitely many halfplanes to define a maximal region of $R(B)$, and that every maximal region of $R(B)$ is convex.


\section{Connected Components}

$B$ is a set of $n$ line segments consisting of $m$ connected components $B_1, \dots, B_m$. Further, $Conv(B_i)$ is the convex hull of the connected component $B_i$. Then for some point $p \in \mathds{R}^2$, we define $L_p(B_i)$ as follows:

\begin{enumerate}
\item If $B_i$ is a single line segment, and $p$ is collinear to $B_i$, then $L_p(B_i) = \emptyset$ 
\item Otherwise, if $p$ lies on a vertex of $Conv(B_i)$, then $L_p(B_i)$ is the double-wedge defined by the lines of the two edges of $Conv(B_i)$ that meet at $p$. 
\item Otherwise, if $p$ lies inside $Conv(B_i)$, or on its boundary, $\partial Conv(B_i)$, then $L_p(B_i) = \mathds{R}^2$
\item Otherwise, $L_p(B_i)$ is the double-wedge defined by the tangents of $Conv(B_i)$ that pass through $p$. 
\end{enumerate}

\begin{figure}[ht]
  \centering
  \includegraphics[scale=0.60]{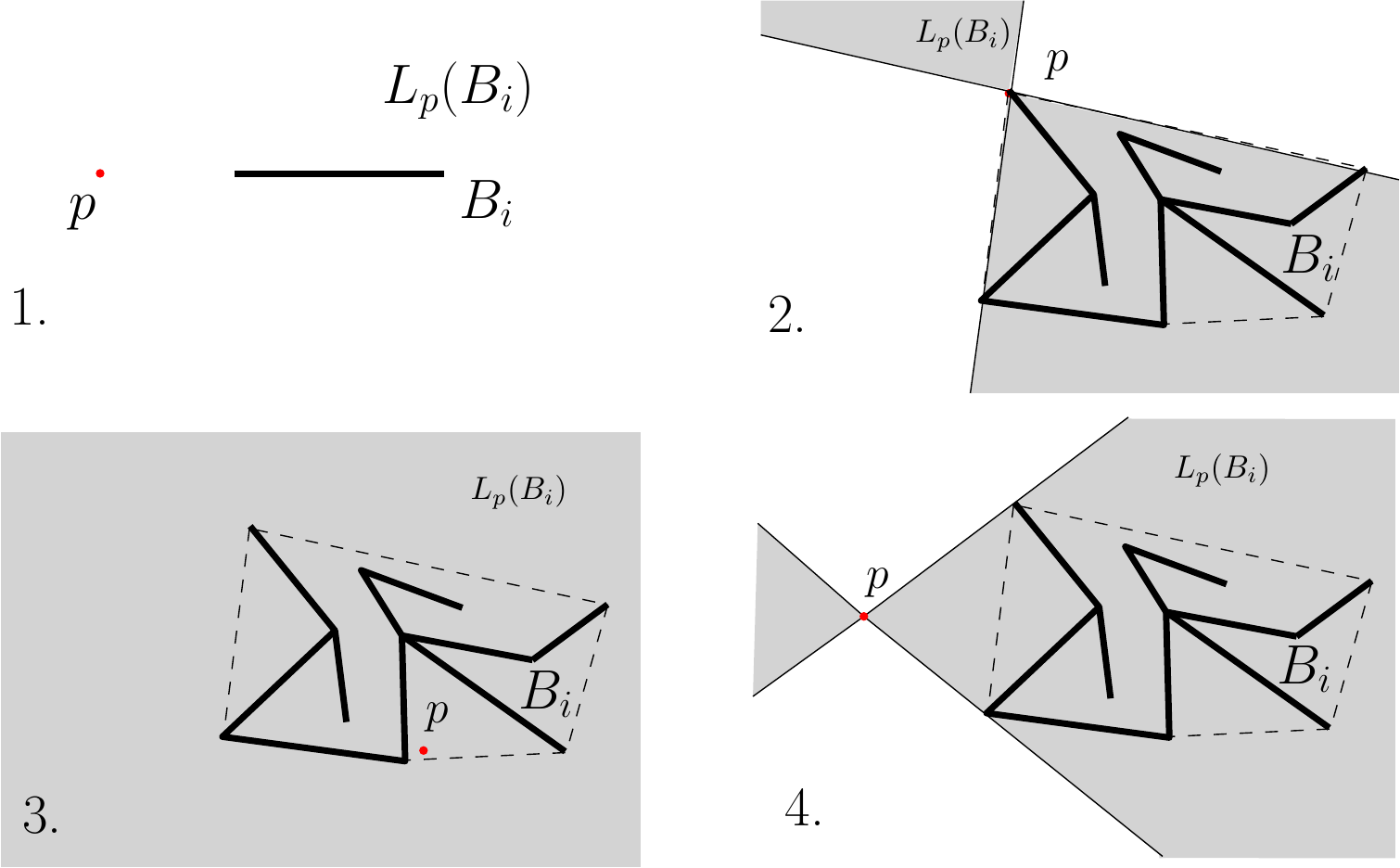}
  \caption{Various possible cases for $L_p(B_i)$}
  \label{fig:pointCoverage}
\end{figure}

Intuitively, $L_p(B_i)$ can be thought of as the set of all lines that pass through $p$ and intersect $B_i$. However this is not strictly true. For parts (3) and (4) of the definition, this does in fact hold. However, (2) describes the limiting behaviour of a point as it tends towards a vertex of $Conv(B_i)$ from outside. (1) ignores the behaviour of points collinear to a single disjoint line segment. This definition may seem counter-intuitive, but it is useful for us. Further, we will consider $L_p(B_i)$ to be a subset of $\mathds{R}^2$, and not the actual lines that pass through $p$. 

\begin{lemma}
Every point in $\overline{L_p(B_i) \cup B_i}$ is a clear point with respect to $B_i$.
\end{lemma}

\begin{proof}
In case (1), where $B_i$ is a single line segment and $p$ is a point collinear with it, this follows trivially, as $R(B_i) = B_i$. Therefore, even though $\overline{L_p(B_i)} = \mathds{R}^2$, the only points that aren't clear are those of $B_i$ itself. 
In case 2, where $p$ is a vertex of $Conv(B_i)$, $Conv(B_i)$ is completely contained within $L_p(B_i)$. 
	Since $R(B_i) = Conv(B_i)$, $\overline{L_p(B_i) \cup B_i}$ can't 
	contain a blocked point. 
In case (3), where $p$ lies inside of or on $Conv(B_i)$ this also follows trivially, as $\overline{L_p(B_i)}$ is empty. 
In case (4), where $p$ lies outside of $Conv(B_i)$, $Conv(B_i)$ is once again completely contained within $L_p(B_i)$ and therefore, $\overline{L_p(B_i) \cup B_i}$ can't contain a blocked point. 
Therefore $\overline{L_p(B_i)}$ is a set of clear points.
\end{proof}

We now define $L_p(B) = \displaystyle\bigcup\limits_{i=1}^{m} L_p(B_i)$ to be the set of all lines which intersect $B$ and pass through $p$, ignoring previously established special cases.

\begin{figure}[ht]
  \centering
  \includegraphics[scale=0.6]{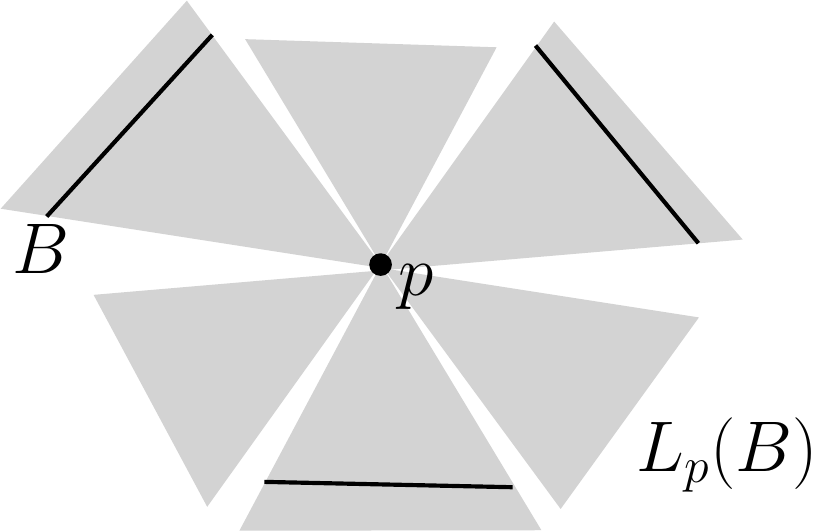}
  \caption{$L_p(B)$ is the set of all lines that intersect $B$ and pass through some point $p$}
  \label{fig:LofQ}
\end{figure} 

Since $\overline{L_p(B_i)\cup B_i}$ is a set of clear points with respect to $B_i$, we can further conclude that $\overline{L_p(B) \cup B}$ has this property with respect to the whole of $B$. Further, for some points $r$ and $s$, since $\overline{L_r(B) \cup B}$ and $\overline{L_s(B) \cup B}$ have this property, $\overline{L_r(B) \cup B} \cup \overline{L_s(B) \cup B}$ also has this property. By DeMorgan's law for set compliments, we can also conclude that $\overline{(L_r(B) \cap L_s(B)) \cup B}$ has this property as well. Therefore given \[L(B) = \displaystyle\bigcap\limits_{i=1}^{m} \ \ \displaystyle\bigcap\limits_{p\text{: vertex of }Conv(B_i)}L_p(B)\] we know $\overline{L(B)\cup B}$ is a set that also has this property. 

\begin{theorem}
\label{thm:coverageComputation}
Let $CI$ be the closure of the interior of a set of points, then $CI(L(B)) \cup B \subseteq R(B) \subseteq L(B) \cup B$. Further, $R(B)\setminus (CI(L(B)) \cup B)$ is a finite set of disjoint points.
\end{theorem}

\begin{proof}
Since $\overline{R(B)}$ is the set of all clear points with respect to $B$, and $\overline{L(B) \cup B}$ is a set of some clear points with respect to $B$, $\overline{R(B)} \supseteq \overline{L(B) \cup B}$. Therefore, $R(B) \subseteq L(B) \cup B$. 

From Lemmas \ref{lemma:hasNoLines} and \ref{lemma:hasPoints}, we know that the only zero area maximal regions of $R(B)$ that aren't in $B$ are individual points. Remark that $CI(L(B))$ differs from $L(B)$ in that only the zero area maximal regions of $L(B)$ have been removed. Therefore, if $CI(R(B)) = CI(L(B))$, $(CI(L(B)) \cup B) \subseteq R(B)$, and all that $R(B)$ and $CI(L(B))$ may differ by are disjoint points. Since $R(B) \subseteq L(B) \cup B$, and $B$ has zero area, $CI(R(B)) \subseteq CI(L(B))$, so all that remains to be proven is $CI(L(B)) \subseteq CI(R(B))$. Equivalently, $\overline{CI(R(B))} \subseteq \overline{CI(L(B))}$

Assume some postive-area region $P$ of points is in 
	$\overline{CI(R(B))}$. Consider a point $p \in P$. There 
	is some line $\ell$ through $p$ that does not intersect $B$. 
	Then $\ell$ can be rotated around $p$ without intersecting $B$ 
	until it is tangent with some connected component $B_i$ at some 
	point $p'$. We will call this rotated line $\ell'$. Now if 
	$p \notin \overline{CI(L(B))}$, then there exists some 
	$L_{p'}(B_j)$, $j\neq i$, which $p$ is in. This would mean 
	there is some line $\ell''$ through $p$ and $p'$ such that 
	$\ell''$ intersects $B_j$. 

\begin{figure}[ht]
  \centering
  \includegraphics[scale=0.4]{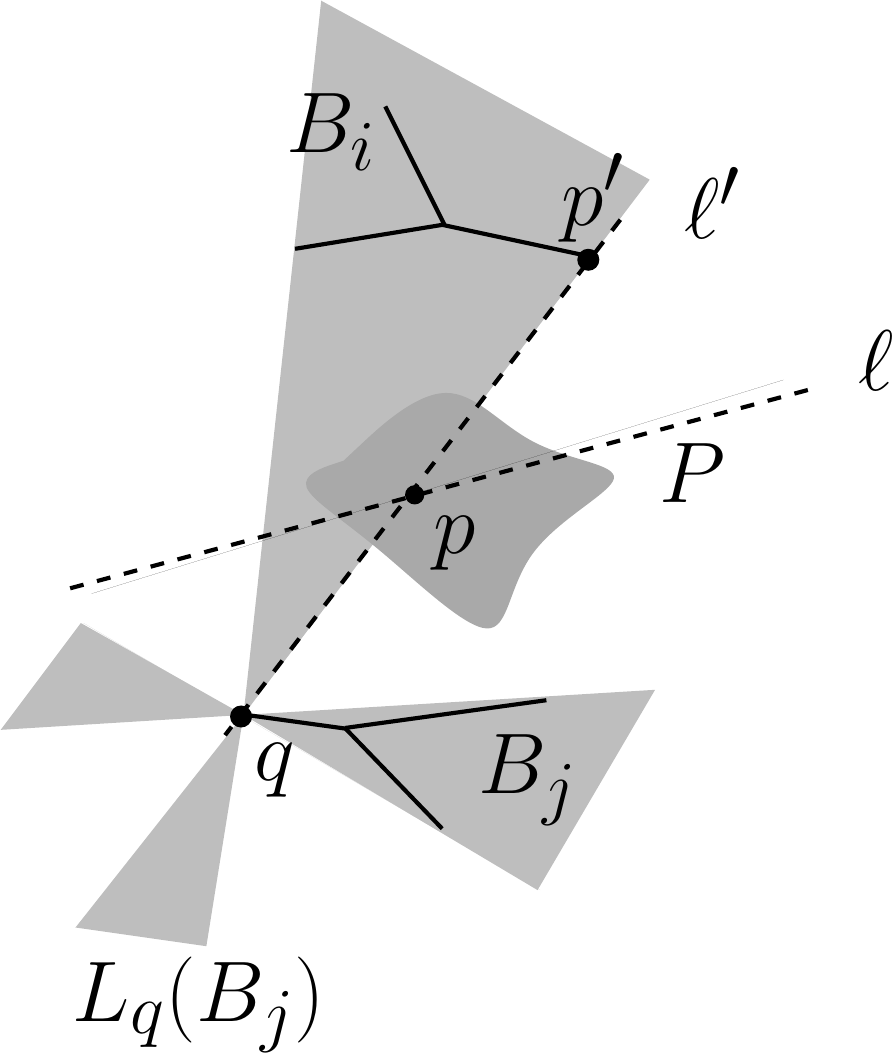}
  \caption{There exists some $L_{p'}(B_j)$, $j\neq i$, which $p$ is in}
  \label{fig:horribleMess}
\end{figure}

However $\ell''$ is $\ell'$, and if $\ell'$ intersects $B_j$ then there are three possibilities. Either $\ell$ intersects $B_j$, we should have stopped at $B_j$ before we got to $B_i$, or $\ell'$ is tangent to $B_j$ as well. For the first two cases we have a contradiction, so $\ell'$ must be tangent to $B_j$. However, since $p$ is part of some region with positive area, we may take a point $p'' \notin \ell'$ adjacent to $p$ such that it lies on no such tangent, and for which this case is therefore not possible. Therefore $p'' \in CI(L(B))$ or else there is a contradiction. Remark that this argument holds for any choice of $p''$ that does not lie on a tangent between two connected components. If the points on these tangents were not in $\overline{CI(L(B))}$ this would imply a region of zero area exists in $CI(L(B))$, but this is impossible. Therefore all points around $p$ must be in $\overline{CI(L(B))}$, and therefore $p$ must be as well. Therefore, if $p \in \overline{CI(R(B))}$, $p \in \overline{CI(L(B))}$, and therefore $CI(L(B)) \subseteq CI(R(B))$.

Since $CI(R(B)) = CI(L(B))$, $R(B)\setminus CI(L(B))$ is a set of disjoint points. To prove that there are finitely many points, recall that by Theorem \ref{thm:tangency-boundary} each maximal region of $R(B)$ is an intersection of halfplanes defined by the vertices of $B$. The only way to get a point from this process is where three or more halfplane boundaries intersect at a point. Since there are finitely many vertices and therefore finitely many halfplanes, it follows that there are finitely many points. 
\end{proof}

\section{Computing the Coverage of a Barrier}

Theorem \ref{thm:coverageComputation} provides a procedure for computing $R(B)$.

We will assume our input is given as a list $B$ of $m$ connected components $B_1, \dots, B_m$, totalling $n$ line segments. The first step of our algorithm will be to compute the convex hulls of all $m$ components. Next, for each vertex $p_k$ of each $Conv(B_i)$, we will compute $L_{p_k}(B_j)$ for each $Conv(B_j)$, and union together these $L_{p_k}(B_j)$ into $L_{p_k}(B)$ by sorting them by angle. Then we will construct an arrangement using all the lines of the $L_{p_k}(B)$. We can then determine our final result by determining how many $L_{p_k}(B)$ one cell is part of, and then traversing the dual while changing our count according to whether a given edge exits or enters an $L_{p_k}(B)$. Then we simply output those regions which were in every $L_{p_k}(B)$, as well as $B$ itself. However this process returns $CI(L(B)) \cup B$, so we may still be missing a finite number of points.

To compute these points, recall that they must lay at the intersection of 3 or more halfplanes. While this is necessary, it is not sufficient. The only way we know of to be certain a point is in $R(B)$ is to perform a radial plane sweep on $B$ from that point. Since there are $O(mn)$ lines in the arrangement, there are $O(m^2n^2)$ candidate points. We will consider a line $\ell$ that makes up the arrangement. There are $O(mn)$ points of intersection on this line. First we will perform a radial plane sweep on one of these points $p$ to construct a set $\Theta = \{\theta_1, \dots, \theta_k\}$ of points on the interval $0$ to $\pi$, where each point $\theta_i$ represents the angle of a tangent to some $B_j$ from $p$, and each point is labelled with the number of connected components the line through $p$ at the angle  $\theta_i + \epsilon$ intersects. If every $\theta_i$ is labelled with a non-zero value, then $p \in R(B)$ and we return it. Now consider the intersection point $q$ on $\ell$ that is adjacent to $p$. While most of the exact values of $\Theta$ will change, the ordering and labelling of the points will only change for those related to the tangents that bound this segment of $\ell$. We can store this data in the vertices of the arrangement during construction, so we can just query $p$ and $q$ for this information. By updating just these values and checking if any are now labelled with $0$, we now know if $q$ is in $R(B)$. Repeating this process for all the points on $\ell$, and then for all choices of $\ell$, we will have determined all the points in $R(B)$.

Computing the convex hulls will take $O(n \log n)$ time. Computing $L_{p_k}(B_j)$ requires computing two lines. In all the special cases this takes constant time, however in the case where we must actually compute the lines as tangents, we take $O(\log n)$ time to binary search $Conv(B_j)$'s vertices for the most extreme points. Since there are $O(m)$ $B_j$, it takes $O(m \log n)$ time to compute them all for one $p_k$. Further, to union them together into $L_{p_k}(B)$, we need to sort their lines by angle, which will take $O(m \log m)$ time. Since there are $O(n)$ $p_k$, we take $O(nm(\log n +\log m))$ time to compute them all. Since we now have $O(nm)$ lines from all our $L_{p_k}(B)$, our arrangement will take $O(m^2n^2)$ time to compute, whose dual we can navigate in $O(m^2n^2)$ time.\cite{Berg08}

For each line of the arrangement we take at most $O(m^2)$ time to perform the plane sweep of the first point. Then for each other point, there are an amortized $O(1)$  other lines intersecting at this point, and we do $O(1)$ work per intersection, so we do $O(mn)$ work per line. Therefore this step takes $O(m^2n^2)$ time. 

Therefore our algorithm runs in $O(m^2n^2)$ time. Now we must determine whether this is good or not.

Since $m$ is at most $n$, our algorithm will run in $O(n^4)$ time in the worst case. Consider the following barrier: Take a regular $n$-gon, and shrink all the edges by a small amount, so that there are gaps where the vertices were. Now there are small regions of space where lines can travel between  each pair of vertices. These regions are equivalent to the planar embedding of $K_n$. This partitions the space into $\Theta(n^4)$ convex regions \cite{Freeman76}. So to even \emph{write} the output it would take $\Omega(n^4)$ time and space. Therefore, our algorithm is indeed worst-case optimal.

\begin{figure}[ht]
  \centering
  \includegraphics[scale=0.4]{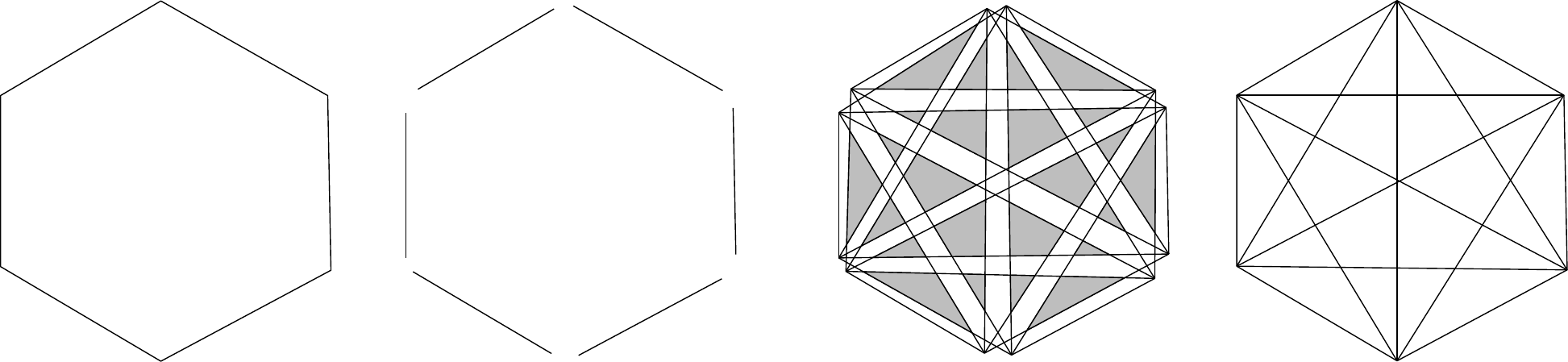}
  \caption{The worst known case barrier and its coverage}
  \label{fig:worstCase}
\end{figure}

\section{Deciding Whether a Point is Part of a Barrier's Coverage}

Given a barrier $B$ one can fairly simply determine whether a point $p$ is in $R(B)$ in $O(n\log n)$ time and $O(n)$ space using a plane sweep. However if $R(B)$ is already constructed, point queries can be done in $O(\log k)$ time using a structure that takes $O(k^2)$ extra space and $O(k^2 \log k)$ time to construct \cite{Kirkpatrick83}, where $k$ is the number of edges in $R(B)$.

\addcontentsline{toc}{chapter}{References}
\small 
\bibliographystyle{abbrv}
\bibliography{bibliography}

\end{document}